\documentclass[10pt,conference]{IEEEtran}  

\usepackage{amssymb,color,amsmath,paralist}
\makeatletter
\def\ps@headings{%
\def\@oddhead{\mbox{}\scriptsize\rightmark \hfil \thepage}%
\def\@evenhead{\scriptsize\thepage \hfil \leftmark\mbox{}}%
\def\@oddfoot{}%
\def\@evenfoot{}}
\makeatother 

\usepackage{supertabular}
\newtheorem{theorem}{\textbf{Theorem}}

\newtheorem{corollary}[theorem]{\textbf{Corollary}}

\newtheorem{definition}[theorem]{\textbf{Definition}}
\newtheorem{lemma}[theorem]{\textbf{Lemma}}
\newtheorem{example}[theorem]{\textbf{Example}}

\DeclareMathOperator{\tr}{tr}
\DeclareMathOperator{\swt}{swt}
\DeclareMathOperator{\wt}{wt}

\newcommand{\C}{\mathbf{C}}
\newcommand{\F}{\mathbf{F}}
\newcommand{\ket}[1]{|#1\rangle}

\newcommand{\acal}[2]{\langle #1\mid #2\rangle_a}
\newcommand{\sdual}{{\perp_s}}
\newcommand{\hdual}{{\perp_h}}
\newcommand{\adual}{{\perp_a}}

\begin{document}

\title{Propagation Rules of Subsystem Codes}
\author{
\authorblockN{Salah A. Aly}
\authorblockA{Department  of Computer Science \\
Texas A\&M University\\
College Station, TX 77843, USA \\
Email: salah@cs.tamu.edu}
 } \maketitle
\begin{abstract}We demonstrate propagation rules of subsystem code constructions by
extending, shortening and combining given subsystem codes. Given an
$[[n,k,r,d]]_q$ subsystem code, we drive new subsystem codes with
parameters $[[n+1,k,r,\geq d]]_q$, $[[n-1,k+1,r,\geq d-1]]_q$,
$[[n,k-1,r+1,d]]_q$. We present the short subsystem codes. The
interested readers shall consult our companion papers  for upper and
lower bounds on subsystem codes parameters, and introduction,
 trading dimensions, families, and references on
subsystem codes~\cite{aly08f,aly08a,aly06c} and references therein.
\end{abstract}

\bigbreak

\noindent \textbf{Subsystem Codes.}
 Let $\mathcal{H}$ be the Hilbert space $C^{q^n}=\C^q \otimes \C^q \otimes ...
\otimes \C^q$. Let $Q$ be a quantum code  such that
$\mathcal{H}=Q\oplus Q^\perp$, where $Q^\perp$ is the orthogonal
complement of $Q$. Recall  definition of the error model acting in
qubits~\cite{calderbank98,aly06c}. We can define the subsystem code
$Q$ as follows
\begin{definition}
An $[[n,k,r,d]]_q$ subsystem code is a decomposition of the subspace
$Q$ into a tensor product of two vector spaces A and B such that
$Q=A\otimes B$, where  $\dim A=k$ and $\dim B= r$. The code $Q$ is
able to detect all errors  of weight less than $d$ on subsystem $A$.
\end{definition}

Subsystem codes can be constructed  from the classical codes  over
$\F_q$ and $\F_{q^2}$. The Euclidean construction of subsystem code
is given as follows~\cite{aly08f,aly06c}.
\begin{lemma}[Euclidean Construction]\label{lem:css-Euclidean-subsys}
If $C$ is a $k'$-dimensional $\F_q$-linear code of length $n$ that
has a $k''$-dimensional subcode $D=C\cap C^\perp$ and $k'+k''<n$,
then there exists an
$$[[n,n-(k'+k''),k'-k'',\wt(D^\perp\setminus C)]]_q$$
subsystem code.
\end{lemma}

\section{Subsystem Codes vers. Co-subsystem Codes}\label{sec:tradingdimensions}
In this section we show how one can trade the dimensions of
subsystem and co-subsystem to obtain new codes from a given
subsystem or stabilizer code. The results are obtained by exploiting
the symplectic geometry of the space. A remarkable consequence is
that nearly any stabilizer code yields a series of subsystem codes.

Our first result shows that one can decrease the dimension of the
subsystem and increase at the same time the dimension of the
co-subsystem while keeping or increasing the minimum distance of the
subsystem code.

\begin{theorem}\label{th:shrinkK}
Let $q$ be a power of a prime~$p$. If there exists an
$((n,K,R,d))_q$ subsystem code with $K>p$ that is pure to $d'$, then
there exists an $((n,K/p,pR,\geq d))_q$ subsystem code that is pure
to $\min\{d,d'\}$. If a pure $((n,p,R,d))_q$ subsystem code exists,
then there exists a $((n,1,pR,d))_q$ subsystem code.
\end{theorem}
\begin{proof}
See~\cite{aly08f,aly08a}
\end{proof}

Replacing $\F_p$-bases by $\F_q$-bases in the proof of the previous
theorem yields the following variation of the previous theorem for
$\F_q$-linear subsystem codes.
\begin{theorem}\label{th:FqshrinkR}
Let $q$ be a power of a prime $p$. If there exists a pure
$\F_q$-linear $[[n,k,r,d]]_q$ subsystem code with $r>0$, then there
exists a pure $\F_q$-linear $[[n,k+1,r-1,d]]_q$ subsystem code.
\end{theorem}
\begin{proof}
See~\cite{aly08f,aly08a}
\end{proof}
\goodbreak
\begin{theorem}[Generic methos]\label{cor:generic}
If there exists an ($\F_q$-linear) $[[n,k,d]]_q$ stabilizer code
that is pure to $d'$, then there exists for all $r$ in the range
$0\le r<k$ an ($\F_q$-linear) $[[n,k-r,r,\ge d]]_q$ subsystem code
that is pure to $\min\{d,d'\}$ .  If a pure ($\F_q$-linear)
$[[n,k,r,d]]_q$ subsystem code exists, then a pure ($\F_q$-linear)
$[[n,k+r,d]]_q$ stabilizer code exists.
\end{theorem}
\begin{proof}
See~\cite{aly08f,aly08a}
\end{proof}
Using this theorem we can derive many families of subsystem codes
derived from families of stabilizer codes as shown in
Table~\ref{table:families}
\begin{figure*}[t]

\begin{center}
\begin{tabular}{|@{}c@{}||c|@{}c@{}|@{}c@{}|}
\hline\hline
Family & Stabilizer $[[n,k,d]]_q$ & Subsystem  $[[n,k-r,r,d]]_q$, $k>r\geq 0$ \\
\hline\hline
Short MDS & $[[n,n-2d+2,d]]_q$  & $[[n,n-2d+2-r,r,d]]_q$ \\
\hline
Hermitian Hamming & $[[n,n-2m,3]]_q$ & $m\ge 2$,  $[[n,n-2m-r,r,3]]_q$ \\
\hline
 Euclidean Hamming & $[[n,n-2m,3]]_q$ & $[[n,n-2m-r,r,3]]_q$ \\

\hline Melas  &$[[n,n-2m,\ge 3]]_q$  & $[[n,n-2m-r,r,\ge
3]]_q$\\
\hline Euclidean BCH & $[[n,n-2m\lceil(\delta-1)(1-1/q)\rceil,\geq \delta]]_q$& $[[n,n-2m\lceil(\delta-1)(1-1/q)\rceil-r,r,\geq \delta]]_q$\\
\hline Hermitian BCH & $[[n,n-2m\lceil(\delta-1)(1-1/q^2)\rceil,\geq
\delta]]_q$ &
 $[[n,n-2m\lceil(\delta-1)(1-1/q^2)\rceil-r,r,\geq \delta]]_q$
  \\
\hline Punctured MDS & $[[q^2-q\alpha, q^2-q\alpha-2\nu-2,\nu+2]]_q$  & $[[q^2-q\alpha, q^2-q\alpha-2\nu-2-r,r,\nu+2]]_q$\\
\hline Euclidean MDS & $[[n,n-2d+2]]_q$& $[[n,n-2d+2-r,r]]_q$\\
\hline Hermitian MDS & $[[q^2-s,q^2-s-2d+2,d]]_q$& $[[q^2-s,q^2-s-2d+2-r,r,d]]_q$\\
\hline
Twisted & $[[q^r,q^r-r-2,3]]_q$ &$[[q^r,q^r-r-2-r,r,3]]_q$\\
\hline Extended twisted & $[[q^2+1,q^2-3,3]]_q$& $[[q^2+1,q^2-3-r,r,3]]_q$\\
\hline Perfect & $[[n,n-s-2,3]]_q$& $[[n,n-s-2-r,r,3]]_q$
\\
&$[[n,n-s-2,3]]_q$& $[[n,n-s-2-r,r,3]]_q$\\ \hline\hline
\end{tabular}
\caption{Families of subsystem codes  from stabilizer
codes}\label{table:families}
\end{center}

\end{figure*}
\section{Propagation Rules}
 Let $C_1 \le \F_q^n$ and $C_2 \F_q^n$ be two classical codes
defined over $F_q$. The direct sum of $C_1$ and $C_2$ is a code $C
\le \F_q^{2n}$ defined as follows

\begin{eqnarray}
C=C_1 \oplus C_2=\{uv \mid u \in C_1, v \in C_2\}.
\end{eqnarray}
In a matrix form the code $C$ can be described as
$$C = \Big(\begin{array}{cc} C_1&0\\0&C_2 \end{array}\Big)$$

An $[n,k_1,d_1]_q$ classical code $C_1$ is a subcode in an
$[c,k_2,d_2]_q$ if every codeword $v$ in $C_1$ is also a codeword in
$C_2$, hence $k_1\leq k_2$. We say that an $[[n,k_1,r_1,d_1]]_q$
subsystem code $Q_1$ is a subcode in an $[[n,k_2,r_2,d_2]]_q$
subsystem code $Q_2$ if  every codeword $\ket{v}$ in $Q_1$ is also a
codeword in $Q_2$ and $k_1+r_1 \leq k_2+r_1$.

{\em Notation.} Let $q$ be a power of a prime integer $p$. We denote
by $\F_q$ the finite field with $q$ elements. We use the notation
$(x|y)=(x_1,\dots,x_n|y_1,\dots,y_n)$ to denote the concatenation of
two vectors $x$ and $y$ in $\F_q^n$. The symplectic weight of
$(x|y)\in \F_q^{2n}$ is defined as $$\swt(x|y)=\{(x_i,y_i)\neq
(0,0)\,|\, 1\le i\le n\}.$$ We define $\swt(X)=\min\{\swt(x)\,|\,
x\in X, x\neq 0\}$ for any nonempty subset $X\neq \{0\}$ of
$\F_q^{2n}$.

The trace-symplectic product of two vectors $u=(a|b)$ and
$v=(a'|b')$ in $\F_q^{2n}$ is defined as
$$\langle u|v \rangle_s = \tr_{q/p}(a'\cdot b-a\cdot b'),$$ where
$x\cdot y$ denotes the dot product and $\tr_{q/p}$ denotes the trace
from $\F_q$ to the subfield $\F_p$.  The trace-symplectic dual of a
code $C\subseteq \F_q^{2n}$ is defined as $$C^\sdual=\{ v\in
\F_q^{2n}\mid \langle v|w \rangle_s =0 \mbox{ for all } w\in C\}.$$
We define the Euclidean inner product $\langle x|y\rangle
=\sum_{i=1}^nx_iy_i$ and the Euclidean dual of $C\subseteq \F_{q}^n$
as $$C^\perp = \{x\in \F_{q}^n\mid \langle x|y \rangle=0 \mbox{ for
all } y\in C \}.$$ We also define the Hermitian inner product for
vectors $x,y$ in $\F_{q^2}^n$ as $\langle x|y\rangle_h
=\sum_{i=1}^nx_i^qy_i$ and the Hermitian dual of $C\subseteq
\F_{q^2}^n$ as
$$C^\hdual= \{x\in \F_{q^2}^n\mid \langle x|y \rangle_h=0 \mbox{ for all } y\in
C \}.$$

\begin{theorem}\label{th:oqecfq}
Let $C$ be a classical additive subcode of\/ $\F_q^{2n}$ such that
$C\neq \{0\}$ and let $D$ denote its subcode $D=C\cap C^\sdual$. If
$x=|C|$ and $y=|D|$, then there exists a subsystem code $Q= A\otimes
B$ such that
\begin{compactenum}[i)]
\item $\dim A = q^n/(xy)^{1/2}$,
\item $\dim B = (x/y)^{1/2}$.
\end{compactenum}
The minimum distance of subsystem $A$ is given by
\begin{compactenum}[(a)]
\item $d=\swt((C+C^\sdual)-C)=\swt(D^\sdual-C)$ if $D^\sdual\neq C$;
\item $d=\swt(D^\sdual)$ if $D^\sdual=C$.
\end{compactenum}
Thus, the subsystem $A$ can detect all errors in $E$ of weight less
than $d$, and can correct all errors in $E$ of weight $\le \lfloor
(d-1)/2\rfloor$.
\end{theorem}

\subsection{\textbf{Extending Subsystem Codes}} We derive new subsystem codes
from known ones by extending and shortening the length of the code.

\begin{theorem}\label{lemma_n+1k}
If there exists an  $((n,K,R,d))_q$ Clifford subsystem code with
$K>1$, then there exists an $((n+1, K, R, \ge d))_q$ subsystem code
that is pure to~1.
\end{theorem}
\begin{proof}
We first note that for any additive subcode $X\le \F_q^{2n}$, we can
define an additive code $X'\le \F_q^{2n+2}$ by
$$X'=\{ (a\alpha|b0)\,|\, (a|b)\in X, \alpha\in
\F_q\}.$$ We have $|X'|=q|X|$. Furthermore, if $(c|e)\in X^\sdual$,
then $(c\alpha|e0)$ is contained in $(X')^\sdual$ for all $\alpha$
in $\F_q$, whence $(X^\sdual)'\subseteq (X')^\sdual$.  By comparing
cardinalities we find that equality must hold; in other words, we
have
$$(X^\sdual)'= (X')^\sdual.$$

By Theorem~\ref{th:oqecfq}, there are two additive codes $C$ and $D$
associated with an $((n,K,R,d))_q$ Clifford subsystem code such that
$$|C|=q^nR/K$$ and $$|D|=|C\cap C^\sdual| = q^n/(KR).$$ We can derive from the
code $C$ two new additive codes of length $2n+2$ over $\F_q$, namely
$C'$ and $D'=C'\cap (C')^\sdual$. The codes $C'$ and $D'$ determine
a $((n+1,K',R',d'))_q$ Clifford subsystem code. Since
\begin{eqnarray*}
D'&=&C'\cap (C')^\sdual = C'\cap (C^\sdual)' \\&=&(C\cap C^\sdual)',
\end{eqnarray*}
 we have
$|D'|=q|D|$. Furthermore, we have $|C'|=q|C|$. It follows from
Theorem~\ref{th:oqecfq} that
\begin{compactenum}[(i)]
\item $K'= q^{n+1}/\sqrt{|C'||D'|}=q^n/\sqrt{|C||D|}=K$,
\item $R'=(|C'|/|D'|)^{1/2} = (|C|/|D|)^{1/2} = R$,
\item $d'= \swt( (D')^\sdual \setminus C')\ge \swt( (D^\sdual\setminus C)')=d$.
\end{compactenum}
Since $C'$ contains a vector $(\mathbf{0}\alpha|\mathbf{0}0)$ of
weight $1$, the resulting subsystem code is pure to~1.
\end{proof}


\begin{corollary}
If there exists an $[[n,k,r,d]]_q$ subsystem  code with $k>0$ and
$0\leq r <k$, then there exists an $[[n+1, k, r, \ge d]]_q$
subsystem code that is pure to~1.
\end{corollary}

\subsection{ \textbf{Shortening Subsystem Codes}} We can also shorten
the length of a subsystem code and still trade the dimensions of the
new subsystem code and its co-subsystem code as shown in the
following Lemma.

\begin{theorem}\label{lem:n-1k+1rule}
If an $((n,K,R,d))_q$ pure subsystem code $Q$ exists, then there is
a pure subsystem code $Q_p$ with parameters $((n-1,qK,R,\geq
d-1))_q$.
\end{theorem}
\begin{proof}
We know that existence of the pure subsystem code $Q$ with
parameters $((n,K,R,d))_q$ implies existence of a pure stabilizer
code with parameters $((n,KR,\geq d))_q$ for $n \geq 2$ and $d\geq
2$ from~\cite[Theorem 2.]{aly08a}. By~\cite[Theorem 70]{ketkar06},
there exist a pure stabilizer code with parameters $((n-1,qKR,\geq
d-1))_q$. This stabilizer code can be seen as $((n-1,qKR,0,\geq
d-1))_q$ subsystem code. By using \cite[Theorem 2.]{aly08a}, there
exists a pure $\F_q$-linear subsystem code with parameters
$((n-1,qK,R,\geq d-1))_q$ that proves the claim.
\end{proof}
Analog of the previous Theorem is the following Lemma.

\begin{lemma}\label{lem:n-1k+1rule}
If an $\F_q$-linear $[[n,k,r,d]]_q$ pure subsystem code $Q$ exists,
then there is a pure subsystem code $Q_p$ with parameters
$[[n-1,k+1,r,\geq d-1]]_q$.
\end{lemma}
\begin{proof}
We know that existence of the pure subsystem code $Q$ implies
existence of a pure stabilizer code with parameters $[[n,k+r,\geq
d]]_q$ for $n \geq 2$ and $d\geq 2$ by using~\cite[Theorem 2. and
Theorem 5.]{aly08a}. By~\cite[Theorem 70]{ketkar06}, there exist a
pure stabilizer code with parameters $[[n-1,k+r+1,\geq d-1]]_q$.
This stabilizer code can be seen as an $[[n-1,k+r+1,0,\geq d-1]]_q$
subsystem code. By using \cite[Theorem 3.]{aly08a}, there exists a
pure $\F_q$-linear subsystem code with parameters $[[n-1,k+1,r,\geq
d-1]]_q$ that proves the claim.
\end{proof}

We can also prove the previous Theorem by defining a new code $C_p$
from the code $C$ as follows.
%
 \begin{theorem}\label{lemma_n-1k+1}
 If there exists a
pure subsystem code $ Q=A\otimes B$ with parameters $((n,K,R,d))_q$
with $n\geq 2$ and $d \geq 2$, then there is a subsystem code $Q_p$
with parameters $((n-1,K,q R, \geq d-1))_q$.
\end{theorem}
\begin{proof}
By Theorem~\ref{th:oqecfq}, if an $((n,K,R,d))_q$ subsystem code $Q$
exists for $K>1$ and $1\leq R<K$, then there exists an additive code
$C \in \F_q^{2n}$ and its subcode $D\leq \F_q^{2n}$ such that
$|C|=q^n R/K$ and $|D|=|C\cap C^{\perp_s}|=q^{n}/KR$. Furthermore,
$d=\min \swt(D^{\perp_s}\backslash C)$. Let $w=(w_1,w_2,\ldots,w_n)$
and $u=(u_1,u_2,\ldots,u_n)$ be two vectors in $\F_q^n$. W.l.g., we
can assume that the code $D^\sdual$ is defined as
$$D^\sdual=\{(u|w) \in \F_q^{2n} \mid w,u \in \F_q^n\}.$$
Let $w_{-1}=(w_1,w_2,\ldots,w_{n-1})$ and
$u_{-1}=(u_1,u_2,\ldots,u_{n-1})$ be two vectors in $\F_q^{n-1}$.
Also, let $D_p^\sdual$ be the code obtained by puncturing the first
coordinate of $D^\sdual$, hence
$$D_p^\sdual=\{(u_{-1}|w_{-1}) \in \F_q^{2n-2} \mid w_{-1},u_{-1} \in \F_q^{n-1}\}.$$
since the minimum distance of $D^\sdual$ is at least 2, it follows
that $|D_p^\sdual|=|D^\sdual|=K^2|C|=K^2q^nR/K=q^nRK$ and the
minimum distance of $D_p^\sdual$ is at least $d-1$. Now, let us
construct the dual code of $D_p^\sdual$ as follows.
\begin{eqnarray*}
(D_p^\sdual)^\sdual &=&\{(u_{-1}|w_{-1}) \in \F_q^{2n-2} \mid \\&& (0u_{-1}|0w_{-1}) \in D,w_{-1},u_{-1} \in
\F_q^{n-1}\}.\end{eqnarray*}

Furthermore, if $(u_{-1}|w_{-1}) \in D_p$, then $(0u_{-1}|0w_{-1})
\in D$. Therefore, $D_p$ is a self-orthogonal code and it has size
given by $$|D_p|=q^{2n-2}/|D_p^\sdual|=q^{n-2}/RK.$$
We can also puncture the code $C$ to the code $C_p$ at the first
coordinate, hence \begin{eqnarray*} C_p&=&\{(u_{-1}|w_{-1}) \in
\F_q^{2n-2} \mid w_{-1},u_{-1} \in \F_q^{n-1},\\ &&
(aw_{-1}|bu_{-1}) \in C, a,b \in F_q\}.\end{eqnarray*} Clearly, $D
\subseteq C$  and if $a=b=0$, then the vector $(0u_{-1}|0w_{-1}) \in
D$, therefore, $(u_{-1},w_{-1}) \in D_p$. This gives us that $D_p
\subseteq C_p$. Furthermore, hence $|C|=|C_p|$. The dual code
$C_p^\sdual$  can be defined as \begin{eqnarray*} C_p^\sdual
&=&\{(u_{-1}|w_{-1}) \in \F_q^{2n-2} \mid w_{-1},u_{-1} \in
\F_q^{n-1},\\ && (ew_{-1}|fu_{-1}) \in C^\sdual, e,f \in
F_q\}.\end{eqnarray*} Also, if $e=f=0$, then $D_p \subseteq
C_p^\sdual$, furthermore, \begin{eqnarray}D_p^\sdual&=&C_p \cup
C_p^\sdual=\{ (u_{-1}|w_{-1}) \in \F_q^{2n-2} \mid \\&&
(0u_{-1}|0w_{-1}) \in D\}\end{eqnarray}

Therefore there exists a subsystem code $Q_p=A_p \otimes B_p$.
Also, the code $D_p^\sdual$ is pure and has minimum distance at
least $d-1$. We can proceed and compute the dimension of subsystem
$A_p$ and co-subsystem $B_p$ from Theorem~\ref{th:oqecfq} as
follows.

\begin{compactenum}[(i)]
\item $K_p= q^{n-1}/\sqrt{|C_p||D_p|}=q^{n-1}/\sqrt{(q^nR/K)(q^{n-2}/RK)}=K$,
\item $R_p=(|C_p|/|D_p'|)^{1/2} = ((q^nR/K)/(q^{n-2}/RK))^{1/2} = qR$,
\item $d_p= \swt( (D_p)^\sdual \setminus C_p)= \swt( (D^\sdual\setminus C_p))\geq d-1$.
\end{compactenum}

  Therefore, there exists a subsystem cod with parameters
$((n-1,K,q R, \geq d-1))_q$.

The minimum distance condition follows since the code $Q$ has $d=
\min \swt (D^{\perp_s}\backslash C)$ and the code $Q_p$ has minimum
distance as $Q$ reduced by one. So, the minimum weight of
$D_p^{\sdual} \backslash C_p$ is at least the minimum weight of
$(D^\sdual \backslash C)-1$
\begin{eqnarray*}
d_p&=&\min \swt (D{_p}^{\perp_s} \backslash C_p)  \nonumber \\
&\geq& \min \swt (D^{\perp_s} \backslash C)-1=d-1
\end{eqnarray*}
If the code $Q$ is pure, then $\min \swt (D^{\perp_s})=d$,
therefore, the new code $Q_p$ is pure since $d_p=\min \swt
(D_p^{\perp_s}) \geq d$.

We conclude that if there is a subsystem code with parameters
$((n-1,K,q R, \geq d-1))_q$, using ~\cite[Theorem 2.]{aly08a}, there
exists a code with parameters $((n-1,qK, R, \geq d-1))_q$.
\end{proof}

\subsection{ \textbf{Reducing Dimension}} We also can reduce dimension
of the subsystem code for fixed length $n$ and minimum distance $d$,
and still obtain a new subsystem code with improved minimum distance
as shown in the following results.

\begin{theorem}\label{lem:reducingK}
If a (pure)$\F_q$-linear $[[n,k,r,d]]_q$ subsystem code $Q$ exists
for $d\geq 2$, then there exists an $\F_q$-linear
$[[n,k-1,r,d_e]]_q$ subsystem code $Q_e$ (pure to d) such that $d_e
\geq d$.
\end{theorem}
\begin{proof}
Existence of the $[[n,k,r,d]]_q$ subsystem code $Q$,  implies
existence of two additive codes $C\leq \F_q^{2n}$ and $D\leq
\F_q^{2n}$ such that $|C|=q^{n-k+r}$ and $|D|=|C \cap
C^{\perp_s}|=q^{n-k-r}$. Furthermore, $d=\min
\swt(D^{\perp_s}\backslash C)$ and $D \subseteq D^\sdual$.

The idea of the proof comes by extending the code $D$ by some
vectors from $D^{\sdual} \backslash (C\cup C^\sdual$). Let us choose
a code $D_e$ of size $|q^{n+1-r-k}|=q|D|$. We also ensure that the
code $D_e$ is self-orthogonal. Clearly extending the code $D$ to
$D_e$ will extend both the codes $C$ and $C^\sdual$ to $C_e$ and
$C_e^\sdual$, respectively. Hence $C_e=q|C|=q^{n+1+r-k}$ and $D_e =
C_e \cap C_e^\sdual$.

There exists a subsystem code $Q_e$ stabilized by the code $C_e$.
The result follows by computing parameters of the subsystem code
$Q_e=A_e \otimes B_e$.

\begin{compactenum}[(i)]
\item $K_e= q^{n}/\sqrt{|C_e||D_e|}=q^{n}/((q^{n+1+r-k})(q^{n+1-k-r}))^{1/2}=q^{k-1}$,
\item $R_e=(|C_e|/|D_e|)^{1/2} = ((q^{n+1}R/K)/(q^{n+1}/RK))^{1/2} = q^r$,
\item $d_e= \swt( (D_e)^\sdual \setminus C_e)\geq \swt( (D^\sdual \setminus C_e))=
d$. If the inequality holds, then the code is pure to $d$.
\end{compactenum}
Arguably, It follows that the set $(D_e^{\sdual}\backslash C_e)$ is
a subset of the set $D^{\sdual}\backslash C$ because $C \leq C_e$,
hence the minimum weight $d_e$ is at least $d$.
\end{proof}


\begin{lemma}\label{lem:reducen-m}
Suppose an $[[n,k,r,d]]_q$ linear pure subsystem code $Q$ exists
generated by the two codes $C,D \leq \F_q^{2n}$. Then there exist
linear $[[n-m,k',r',d']]_q$ and $[[n-m,k'+r'-r'',r'',d']]_q$
subsystem codes with $k'\geq k-m$, $r' \geq r$, $0\leq r''< k'+r'$,
and $d'\geq d$ for any integer $m$ such that there exists a codeword
of weight $m$ in $(D^{\perp_s}\backslash C)$.
\end{lemma}
\begin{proof}{[Sketch]}
This lemma~\ref{lem:reducen-m} can be proved easily by mapping the
subsystem code $Q$ into a stabilizer code. By using \cite[Theorem
7.]{calderbank98}, and the new resulting stabilizer code can be
mapped again to a subsystem code with the required parameters.
\end{proof}

\subsection{ \textbf{Combining Subsystem Codes}} We can also
construct new subsystem codes from given two subsystem codes. The
following theorem shows that two subsystem codes can be merged
together into one subsystem code with possibly improved distance or
dimension.

\begin{theorem}\label{thm:twocodes_n1k1r1d1n2k2r2d2}
Let $Q_1$ and $Q_2$ be two pure subsystem codes with parameters
$[[n_1,k_1,r_1,d_1]]_2$ and $[[n_2,k_2,r_2,d_2]]_2$ for $k_2+r_2\leq
n_1$, respectively. Then there exists a subsystem code with
parameters $[[n_1+n_2-k_2-r_2,k_1+r_1-r,r,d]]_2$, where $d \geq min
\{d_1,d_1+d_2-k_2-r_2\}$ and $0 \leq r <k_1+r_1$.
\end{theorem}

\begin{proof}
Existence of an $[[n_i,k_i,r_i,d_i]]_2$ pure subsystem code $Q_i$
for $i\in\{1,2\}$ , implies existence of a pure stabilizer code
$S_i$ with parameters $[[n_i,k_i+r_i,d_i]]_2$ with $k_2+r_2 \leq
n_1$, see~\cite{aly08a}. Therefore, by~\cite[Theorem
8.]{calderbank98}, there exists  a stabilizer code with parameters
$[[n_1+n_2-k_2-r_2,k_1+r_1,d]]_2$, $d \geq \min
\{d_1,d_1+d_2-k_2-r_2\}$. But this code gives us a subsystem code
with parameters $[[n_1+n_2-k_2-r_2,k_1+r_1-r,r,\geq d]]_2$ with
$k_2+r_2 \leq n_1$ and $0 \leq r <k_1+r_1$ that proves the claim.
\end{proof}

\begin{theorem}\label{lem:twocodes_nk1r1d1nk2r2d2} Let $Q_1$ and $Q_2$ be two
pure subsystem codes with parameters $[[n,k_1,r_1,d_1]]_q$ and
$[[n,k_2,r_2,d_2]]_q$, respectively. If $Q_2\subseteq Q_1$, then
there exists an $[[2n,k_1+k_2+r_1+r_2-r,r,d]]_q$ pure subsystem code
 with minimum distance $d \geq \min \{d_1,2d_2\}$ and $0\leq r <k_1+k_2+r_1+r_2$.
\end{theorem}
\begin{proof}
Existence of a pure subsystem code with parameters
$[[n,k_i,r_i,d_i]]_q$ implies existence of a pure stabilizer code
with parameters $[[n,k_i+r_i,d_i]]_q$ using~\cite[Theorem
4.]{aly08a}. But by using~\cite[Lemma 74.]{ketkar06}, there exists a
pure stabilizer code with parameters $[[2n,k_1+k_2+r_1+r_2,d]]_q$
with $d \geq \min
 \{2d_2,d_1\}$. By~\cite[Theorem 2., Corollary 6.]{aly08a}, there
 must exist a pure subsystem code with parameters
 $[[2n,k_1+k_2+r_1+r_2-r,r,d]]_q$ where $d \geq \min
 \{2d_2,d_1\}$ and $0\leq r <k_1+k_2+r_1+r_2$, which proves the claim.
\end{proof}

We can recall the trace alternative product between  two codewords
of a classical code and the proof
of~Theorem~\ref{lem:twocodes_nk1r1d1nk2r2d2} can be stated as
follows.

\begin{lemma}\label{lem:twocodes_nk1r1d1nk2r2d2another} Let $Q_1$
and $Q_2$ be two pure subsystem codes with parameters
$[[n,k_1,r_1,d_1]]_q$ and $[[n,k_2,r_2,d_2]]_q$, respectively. If
$Q_2\subseteq Q_1$, then there exists an
$[[2n,k_1+k_2,r_1+r_2,d]]_q$ pure subsystem code
 with minimum distance $d \geq \min \{d_1,2d_2\}$.
\end{lemma}

\begin{proof}
Existence of the code $Q_i$ with parameters $[[n,K_i,R_i,d_i]]_q$
implies existence of two additive  codes $C_i$ and $D_i$ for $i \in
\{1,2\}$ such that $ |C_i|=q^nR_i/K_i$ and $|D_i|=|C\cup
C^{\perp_s}|=q^n/R_iK_i$.

 We know that there exist additive linear
codes $D_i \subseteq D_i^\adual$, $D_i \subseteq C_i$, and $D_i
\subseteq C_i^\adual$. Furthermore, $D_i=C_i \cap C_i^\adual$ and
$d_i = wt (D_i^\adual \backslash C_i)$. Also, $C_i = q^{n+r_i-k_i}$
and $|D|=q^{n-r_i-k_i}$.

Using the direct sum definition between to linear codes, let us
construct a code $D$ based on $D_1$ and $D_2$ as

$$D=\{(u,u+v) \mid u \in D_1, v \in D_2\} \le \F_{q^2}^{2n}.$$
The code $D$ has size of $|D|=q^{2n -(r_1+r_2+k_1+k_2)=|D_1||D_2|}.$
Also, we can define the code $C$ based on the codes $C_1$ and $C_2$
as
$$C=\{(a,a+b) \mid a \in C_1, b \in C_2\} \le \F_{q^2}^{2n}.$$
The code $C$ is of size $|C|=|C_1||C_2|=q^{2n+r_1+r_2-k_1-k_2}.$ But
the trace-alternating dual of the code $D$ is

$$D^{\adual}= \{(u'+v'|,v') \mid u' \in D_1^\adual, v' \in D_2^\adual \}.$$
 We notice that $(u'+v',v')$ is orthogonal to $(u,u+v)$ because,
 from properties of the product,
\begin{eqnarray*}\acal{(u,u+v)}{(u'+v',v')}&=&\acal{u}{u'+v'}+
\acal{u+v}{v'}\\ &=& 0\end{eqnarray*} holds for $u\in D_1,v\in
D_2,u'\in D_1^\adual,$ and $v' \in D_2^\adual$.

Therefore, $D \subseteq D^\adual$ is a self-orthogonal code with
respect to the trace alternating product. Furthermore, $C^\adual=
\{(a'+b',b') \mid a' \in C_1^\adual, b' \in C_2^\adual\}.$ Hence,
$C\cap C^\adual = \{(a,a+b) \cap (aa+b',b')\}=D$. Therefore, there
exists an $\F_q$-linear subsystem code $Q=A \otimes B$ with the
following parameters.

\begin{compactenum}[i)]
\item \begin{eqnarray*}K&=&|A|=q^{2n}/(|C||D|)^{1/2}\\&=&
\frac{q^{2n}}{\sqrt{(q^{2n}R_1R_2/K_1K_2)(q^{2n}/K_1K_2R_1R_2)}}\\&=&\frac{q^{2n}}{\sqrt{q^{2n+r_1+r_2-k_1-k_2}q^{2n-r_1-r_2-k_1-k_2}}}\\&=&q^{k_1k_2}=K_1K_2.
\end{eqnarray*}
\item $R=(\frac{|C|}{|D|})^{1/2}=R_1R_2.$
\item the minimum distance is a direct consequence.
\end{compactenum}

\end{proof}
%
%
%

%

\begin{theorem}
If there exist two pure subsystem quantum codes $Q_1$ and $Q_2$ with
parameters $[[n_1,k_1,r_1,d_1]]_q$ and $[[n_2,k_2,r_2,d_2]]_q$,
respectively. Then there exists a pure subsystem code $Q'$ with
parameters $[[n_1+n_2,k_1+k_2+r_1+r_2-r,r,\geq \min (d_1,d_2)]]_q$.
\end{theorem}
\begin{proof}
This Lemma can be proved easily from~\cite[Theorem 5.]{aly08a}
and~\cite[Lemma 73.]{ketkar06}. The idea is to map a pure subsystem
code to a pure stabilizer code, and once again map the pure
stabilizer code to a pure subsystem code.
\end{proof}

\begin{theorem}
If there exist two pure subsystem quantum codes $Q_1$ and $Q_2$ with
parameters $[[n_1,k_1,r_1,d_1]]_q$ and $[[n_2,k_2,r_2,d_2]]_q$,
respectively. Then there exists a pure subsystem code $Q'$ with
parameters $[[n_1+n_2,k_1+k_2,r_1+r_2,\geq \min (d_1,d_2)]]_q$.
\end{theorem}
\begin{proof}
Existence of the code $Q_i$ with parameters $[[n,K_i,R_i,d_i]]_q$
implies existence of two additive  codes $C_i$ and $D_i$ for $i \in
\{1,2\}$ such that $ |C_i|=q^nR_i/K_i$ and $|D_i|=|C\cup
C^{\perp_s}|=q^n/R_iK_i$.

 Let us choose the codes $C$ and $D$ as follows.
 $$C=C_1 \oplus C_2=\{uv \mid v \in C_1, v\in C_2\},$$ and $$D=D_1 \oplus D_2=\{ab \mid a\in D_1, b\in
 C_2\},$$
 respectively. From this construction, and since $D_1$ and $D_2$ are
 self-orthogonal codes, it follows that $D$ is also a
 self-orthogonal code. Furthermore, $D_1 \subseteq C_1$ and $D_2
 \subseteq
 C_2$, then

  $$D_1 \oplus D_2 \subseteq C_1 \oplus C_2,$$
hence $D\subseteq C$.  The code $C$ is of size
\begin{eqnarray*}|C|&=&|C_1||C_2|=q^{(n_1+n_2)-(k_1+k_2)+(r_1+r_2)}\\&=&q^{n_1}q^{n_2}R_1R_2/K_1K_2\end{eqnarray*}
and  $D$ is of size
\begin{eqnarray*}|D|&=&|D_1||D_2|=q^{(n_1+n_2)-(k_1+k_2)-(r_1+r_2)}\\&=&q^{n_1}q^{n_2}/R_1R_2K_1K_2.\end{eqnarray*}
On the other hand,
\begin{eqnarray*}C^\sdual=(C_1 \oplus C_2)^\sdual = C_2^\sdual \oplus C_1^\sdual
\supseteq D_2\oplus D_1.\end{eqnarray*} Furthermore, $C\cap
C^\sdual=(C_1 \oplus C_2)\cap (C_2^\sdual \cap C_1^\sdual)=D$.

Therefore, there exists a subsystem code $Q=A\otimes B$ with the
following parameters.
\begin{compactenum}[i)]
\item \begin{eqnarray*}K&=&|A|=q^{n_1+n_2}/(|C||D|)^{1/2}\\&=&
\frac{q^{n_1+n_2}}{\sqrt{(q^{n_1+n_2}R_1R_2/K_1K_2)(q^{n_1+n_2}/K_1K_2R_1R_2)}}\\&=&\frac{q^{n_1+n_2}}{\sqrt{q^{n_1+n_2+r_1+r_2-k_1-k_2}q^{n_1+n_2-r_1-r_2-k_1-k_2}}}\\&=&q^{k_1k_2}=K_1K_2=|A_1||A_2|.
\end{eqnarray*}
\item \begin{eqnarray*}
R&=&(\frac{|C|}{|D|})^{1/2}=\sqrt{\frac{q^{n_1}q^{n_2}R_1R_2/K_1K_2}{q^{n_1}q^{n_2}/R_1R_2K_1K_2}}\\&=&R_1R_2=|B_1||B_2|.\end{eqnarray*}
\item the minimum weight of $D^{\perp_s} \backslash C$ is at least the
minimum weight of $D_1^{\perp_s}\backslash C_1$ or
$D_2^{\perp_s}\backslash C_2$.
\begin{eqnarray*}
d&=&\min \{\swt (D_1^{\perp_s} \backslash C_1),(D_2^{\perp_s}
\backslash C_2)\} \nonumber \\ &\geq& \min \{d_1,d_2\}.
\end{eqnarray*}
\end{compactenum}

\end{proof}

\begin{theorem}
Given two pure subsystem codes $Q_1$ and $Q_2$ with parameters
$[[n_1,k_1,r_1,d_1]]_q$ and $[[n_2,k_2,r_2,d_2]]_q$, respectively,
with $k_2 \leq n_1$. An $[[n_1+n_2-k_2,k_1+r_1+r_2-r,r,d]]_q$
subsystem code exists such that $d\geq \min \{d_1,d_1+d_2-k_2\}$ and
$0\leq r <k_1+r_1+r_2$.
\end{theorem}
\begin{proof}
The proof is a direct consequence as shown in the previous theorems.
\end{proof}


\begin{theorem} If an $((n,K,R,d))_{q^m}$ pure subsystem code exists, then there exists a
pure subsystem code with parameters $((nm,K,R,\geq d))_q$.
Consequently, if a pure subsystem code with parameters
$((nm,K,R,\geq d))_q$ exists, then there exist a subsystem code with
parameters $((n,K,R,\geq \lfloor d/m \rfloor))_{q^m}$..
\end{theorem}
\begin{proof}
Existence of a pure subsystem code with parameters
$((n,K,R,d))_{q^m}$ implies existence of a pure stabilizer code with
parameters $((n,KR,d))_{q^m}$ using~\cite[Theorem 5.]{aly08a}.
By~\cite[Lemma 76.]{ketkar06}, there exists a stabilizer code with
parameters  $((nm,KR,\geq d))_q$. From~\cite[Theorem 2,5.]{aly08a},
there exists a pure subsystem code with  parameters $((nm,K,R,\geq
d))_q$ that proves the first claim.  By~\cite[Lemma 76.]{ketkar06}
and ~\cite[Theorem 2,5.]{aly08a}, and repeating the same proof, the
second claim is a consequence.
\end{proof}

\begin{figure*}[t]
\begin{center}
\begin{supertabular}{ c@{\hspace{0.2cm}}|| c@{\hspace{0.2cm}}| c@{\hspace{0.2cm}}|c@{\hspace{0.3cm}}}
\hline \hline n $\backslash$ k &k-1 &k &k+1\\
 \hline \hline

 n-1 & $[ r+2, d-1]_q$& $[\leq r+2,d]_q$, $[r+1,d-1]_q$& $[r,d-1]_q$\\
\hline
  n & $[ r+1, d]_q$, $[ r+1,\geq d]_q$&  $[r,d]_q$ $\rightarrow [ \leq r, \geq d]_q$&$[ r-1,  d]_q$\\
  &&$ \rightarrow [\geq r,\leq d]_q $&\\
\hline
   n+1 & $[\geq r,\geq d]_q$&$[\geq r,d]_q$ , $[r,\geq d]_q$&\\
 \hline
\end{supertabular}
\caption{Existence of subsystem propagation rules}
\end{center}
\end{figure*}[t]

\section{Special and Short Subsystem Codes $[[8,1,2,3]]_2$ and $[[6,1,1,3]]_3$}\label{Sec:ShortSubsys}
In this section we present the shortest subsystem codes over $\F_2$
and $\F_3$ fields.  Theorem~\ref{cor:generic} implies that a
stabilizer code with parameters $[[n,k,d]]_q$ gives subsystem codes
with parameters $[[n,k-r,r,d]]_q$, see the tables in~\cite{aly08f}.

Consider a stabilizer code with parameters $[[8,3,3]]_2$. This code
can be used to derive $[[8,2,1,3]]_2$ and $[[8,1,2,3]]_2$ subsystem
codes. We give an explicit construction of these codes. Further, we
claim that $[[8,1,2,3]]_2$ and $[[8,2,1,3]]_2$ are the shortest
nontrivial binary subsystem codes.  We show the stabilizer and
normalizer matrices for these codes. Also, we prove their minimum
distances using the weight enumeration of  these codes. We present
two codes with less length, however we can not tolerate more than 2
gauge qubits. The following example shows $[[8,1,2,3]]$ short
subsystem code over $\F_2$.
\begin{example}

\begin{eqnarray}
D_S= \left[ \begin{array} {cccccccc}
X & I & Y & I & Z & Y & X & Z \\
Y & I & Y & X & I & Z & Z & X \\
I & X & Y & Y & Z & X & Z & I \\
I & Y & I & Z & Y & X & X & Z \\
I & I & X & Z & X & Y & Z & Y \\
\end{array} \right]
\end{eqnarray}

\begin{eqnarray}
D^\perp_S= \left[ \begin{array} {cccccccc}
 X & I & I & I & I & I & Z & Y\\
 Y & I & I & I & I & Y & X & X \\
 I & X & I & I & I & Y & Y & X \\
 I & Y & I & I & I & I & X & Z \\
 I & I & X & I & I & Y & Z & I\\
 I & I & Y & I & I & I & Z & X \\
 I & I & I & X & I & Y & I & Z \\
 I & I & I & Y & I & Y & Y & Y\\
 I & I & I & I & X & I & Y & Z\\
 I & I & I & I & Y & Y & Z & Z\\
 I & I & I & I & I & Z & X & Y\\
\end{array} \right]
\end{eqnarray}

\begin{eqnarray}
C_S= \left[ \begin{array} {cccccccc}
X & I & Y & I & Z & Y & X & Z \\
Y & I & Y & X & I & Z & Z & X \\
I & X & Y & Y & Z & X & Z & I \\
I & Y & I & Z & Y & X & X & Z \\
I & I & X & Z & X & Y & Z & Y \\
\hline
 Y & I & I & I & I & Y & X & X \\
 I & X & I & I & I & Y & Y & X \\
\end{array} \right]
\end{eqnarray}

\begin{eqnarray}
C^\perp_S= \left[ \begin{array} {cccccccc}
X & I & Y & I & Z & Y & X & Z \\
Y & I & Y & X & I & Z & Z & X \\
I & X & Y & Y & Z & X & Z & I \\
I & Y & I & Z & Y & X & X & Z \\
I & I & X & Z & X & Y & Z & Y \\
\hline
X & I & I & I & I & I & Z & Y\\
 I & I & I & Y & I & Y & Y & Y\\
\end{array} \right]
\end{eqnarray}

We notice that the matrix $D_S$ generates the code $D=C \cap
C^{\perp_s}$. Furthermore, dimensions of the subsystems $A$ and $B$
are given by $k=\dim D^{\perp_s}- \dim C=(11-7)/2=2$ and $r= \dim C
- \dim D=(7-5)/2=1$. Hence we have $[[8,2,1,3]]_2$ and
$[[8,1,2,3]]_2$ subsystem codes.
\end{example}

We show that the subsystem codes $[[8,1,2,3]]_2$ is not better than
the stabilizer code $[[8,3,3]]_2$ in terms of syndrome measurement.
The reason is that the former needs $8-1-2=5$ syndrome measurements,
while the later needs also $8-3=5$ measurements. This is an obvious
example where subsystem codes have no superiority in terms of
syndrome measurements.

We post an open question regarding the threshold value and fault
tolerant gate operations for this code. We do not know at this time
if the code $[[8,1,2,3]]_2$ has better threshold value and less
fault-tolerant operations. Also, does the subsystem code with
parameters  $[[8,1,3,3]]_2$ exist?

\textbf{No nontrivial $[[7,1,1,3]]_2$ exists.} There exists a
trivial $[[7,1,1,3]]_2$ code obtained by simply extending the
$[[7,1,3]]_2$ code as the $[[5,1,3]]_2$ code. We show the smallest
subsystem code with length $7$ must have at most  minimum weight
equals to 2. Since $[[7,2,2]]_2$ exists, then we can construct the
stabilizer and normalizer matrices as follows.

\begin{eqnarray}
D_S= \left[ \begin{array} {ccccccc}
 X & X & X & X & I & I & I\\
 Y & Y & Y & Y & I & I & I \\
 I & I & I & I & X & I & I \\
 I & I & I & I & I & X & I \\
 I & I & I & I & I & I & X\\
\end{array} \right]
\end{eqnarray}

\begin{eqnarray}
D^\perp_S= \left[ \begin{array} {ccccccc}
X & I & I & X & I & I & I \\

Y & I & I & Y & I & I & I \\

I & X & I & X & I & I & I \\

I & Y & I & Y & I & I & I \\

I & I & X & X & I & I & I \\

I & I & Y & Y & I & I & I \\

I & I & I & I & X & I & I\\

I & I & I & I & I & X & I\\

I & I & I & I & I & I & X\\

\end{array} \right]
\end{eqnarray}
Clearly, from our construction and using Theorem~\ref{cor:generic},
there must exist a subsystem code with parameters $k$ and $r$ given
as follows. $\dim D^{\perp_s}=9/2$ and $\dim C=7/2$. Also, $\dim
D=5/2$ and $\ min (D^{\perp_s} \backslash C)=2$. Therefore, ,
$k=(9-7)/2=1$ and $r=(7-5)/2=1$. Consequently, the parameters of the
subsystem code are $[[7,1,1,2]]_2$.

\smallskip

This example shows $[[6,1,1,3]]$ short subsystem code over $\F_3$.
\begin{example}
We give a nontrivial short subsystem code over $\F_3$. This is
derived from the $[[6,2,3]]_3$ graph quantum code. Also, we show
in~\cite{aly08f} an example  for an $[[6,1,1,3]]_7$ subsystem code
over $\F_7$. Consider the field $\F_3$ and let $C\subseteq
\F_3^{12}$ be a linear code defined by the following generator
matrix.
\begin{eqnarray*}
C=\left[ \begin{array}{rrrrrr|rrrrrr}
1&0&0&0&2&0&0&2&0&2&0&2\\
0&1&0&0&0&2&1&0&1&0&1&0\\
0&0&1&0&2&0&0&1&0&1&0&1\\
0&0&0&1&0&2&2&0&2&0&2&0\\
\hline
0&0&0&0&1&0&0&2&0&1&0&0\\
0&0&0&0&0&0&1&0&1&0&1&0\\
\end{array}\right]=\left[\begin{array}{c} S\\ \hline X_1\\Z_1\\
\end{array}\right].
\end{eqnarray*}
Let the symplectic inner product $\langle (a|b)|(c|d)\rangle_s =
a\cdot d - b\cdot c$. Then the symplectic dual of $C$ is generated
by
\begin{eqnarray*}
C^\sdual=\left[\begin{array}{c}  S\\ \hline X_2\\Z_2\\
\end{array}\right],
\end{eqnarray*}
where $X_2=\big[ \begin{array}{rrrrrr|rrrrrr}
0&0&0&0&0&1&1&0&2&0&0&0\\
\end{array} \big]$ and\\ $Z_2=\big[ \begin{array}{rrrrrr|rrrrrr}
0&0&0&0&0&0&0&1&0&1&0&1\\
\end{array} \big]$. The matrix $S$ generates the code $D=C\cap
C^\sdual$. Now $D$ defines a $[[6,2,3]]_3$ stabilizer code.
Therefore, $\swt(D^\sdual\setminus D)=3$. It follows that
$\swt(D^\sdual\setminus C) \geq \swt(D^\sdual)=3$. By
\cite[Theorem~4]{aly06c}, we have a $[[6,(\dim D^\sdual-\dim
C)/2,(\dim C-\dim D)/2, 3]]_3$ viz. a $[[6,1,1,3]]_3$ subsystem
code.
\end{example}

We can also have a trivial $[[6,1,1,3]]_2$ code.  This trivial
extension seems to argue against the usefulness of subsystem codes
and if they will really lead to improvement in performance. An
obvious open question is if there exist nontrivial $[[6,1,1,3]]_2$
or $[[7,1,1,3]]_2$ subsystem codes.

\newcommand{\XXstud}{{}}
\newcommand{\XXar}[1]{}
\bibliographystyle{plain}

\end{document}